\pdfoutput=1

\documentclass[conference,letterpaper]{article}

\usepackage{appendix}
\usepackage[utf8]{inputenc} 
\usepackage[T1]{fontenc}
\usepackage{url}
\usepackage{ifthen}
\usepackage{amsthm}
\usepackage{cite}
\usepackage{amsfonts}
\usepackage[cmex10]{amsmath} 

\newtheorem{theorem}{Theorem}
\newtheorem{lemma}{Lemma}
\newtheorem{corr}{Corollary}
\newtheorem{fact}{Fact}
\newtheorem{definition}{Definition}
\newcommand{\ver}[1]{v(\{#1\})}
\newcommand{\nalph}{n^\alpha}
\newcommand{\omnalph}{n^{1-\alpha}}
\begin{document}
\title{Age of Gossip in Random and Bipartite Networks} 



\if{false}
\author{%
  \IEEEauthorblockN{Thomas J. Maranzatto}
  \IEEEauthorblockA{Department of Mathematics, Statistics, and Computer Science\\
                    University of Illinois Chicago\\
                    Chicago, IL, USA\\
                    Email: tmaran2@uic.edu}
}
\fi
\author{%
  Thomas J. Maranzatto\\
  Department of Mathematics, Statistics, and Computer Science\\
                    University of Illinois Chicago\\
                    Chicago, IL, USA\\
                    Email: tmaran2@uic.edu
}

\maketitle


\begin{abstract}
In this paper we study gossip networks where a source observing a process sends updates to an underlying graph.  Nodes in the graph communicate to their neighbors by randomly sending updates.  Our interest is studying the version age of information (vAoI) metric over various classes of networks.  It is known that the version age of $K_n$ is logarithmic, and the version age of $\overline{K_n}$ is linear.  We study the question `how does the vAoI evolve as we interpolate between $K_n$ and $\overline{K_n}$' by studying Erd\H{o}s-Reyni random graphs, random $d$-regular graphs, and bipartite networks.   Our main results are proving the existence of a threshold in $G(n,p)$ from rational to logarithmic average version age, and showing $G(n,d)$ almost surely has logarithmic version age for constant $d$.  We also characterize the version age of complete bipartite graphs $K_{L,R}$, when we let $L$ vary from $O(1)$ to $O(n)$.
\end{abstract}

\section{Introduction}\label{intro}
We live in a world of connectivity, and randomized communication networks are becoming more ubiquitous.  From IoT applications~\cite{random_iot} to federated machine learning topologies~\cite{random_federated} to connectivity in drone swarms~\cite{random_drone}, it is clear that analysis of information flow in random networks is needed.

In this paper we consider the version age of information (vAoI) metric over various gossip network topologies.  Herein, a source node observing a process randomly shares its status with a network of $\mathcal{N}$ nodes, and the nodes in the network randomly share (or `gossip') their statuses amongst each other. A central question in gossip networks is to characterize how out of date the status of the network is compared to the source, and the vAoI is just one among many metrics used in this analysis.  Others ways to measure `freshness' include the age of information \cite{age_information}, age of incorrect information \cite{age_incorrect}, and age of synchronization \cite{age_sync}.

The vAoI metric was initially used to show that the age of gossip in the clique $K_n$ scales as $\Theta(\log n)$, and the age of gossip in the disconnected graph $\overline{K_n}$ scales as $\Theta(n)$~\cite{yates_gossip}.
 A conjecture was also given that the cycle $C_n$ should have version age $O(\sqrt{n})$ which was subsequently proven\cite{gossip_community}. Other works have studied the 2-d lattice \cite{gossip_grid} and generalizations of cycles \cite{gossip_rings}.

 This paper investigates the question, `how does the vAoI evolve as we interpolate between $K_n$ and $\overline{K_n}$'?  This is a natural question to ask, and was stated in other works~\cite{gossip_grid} and partially answered in the context of generalized cycles~\cite{gossip_rings}.  We choose to focus on complete bipartite graphs $K_{L,R}$ and random graph models.  For bipartite graphs, we investigate how version age increases when the left partition increases from $0$ to $\frac{n}{2}$.  For random graphs, we show a phase transition takes place where the graph goes from having rational version age to logarithmic version age.  We also study random d-regular graphs and show under the uniform distribution, these almost surely always have logarithmic version age.   Our main theorems are presented in section~\ref{sec:summary}.

\section{System Model and Background}\label{sec:model}
\subsection{Version Age of Information}
We consider a source node $n_0$ sending updates to a network $G = (\mathcal{N}, E)$ over $n$ nodes.  We let $\mathcal{N} = [1,..,n]$.  The source updates itself via a Poisson process with rate $\lambda_e$. The source also sends updates to each $v \in G$ as separate Poisson processes with rates $\frac{\lambda}{n}$. In all graphs we consider, an undirected edge $ij \in E$ facilitates two-way communication between nodes $i$ and $j$, with $\lambda_i(j) = \frac{\lambda}{\operatorname{deg}(i)}$ denoting the Poisson rate from $i$ to $j$, and $\lambda_j(i) = \frac{\lambda}{\operatorname{deg}(j)}$ denoting the Poisson rate from $j$ to $i$.  In general $\lambda_i(j) \not= \lambda_j(i)$.  If node $i$ has no neighbors then for all $j$, $\lambda_i(j) = 0$.

The source and every node in the network have internal counters; when a node $i \in \mathcal{N}\cup \{n_0\}$ communicates to a neighbor $j$ (because $i$'s Poisson process updated) $i$ sends its current counter value. The counter for $n_0$ increments if and only if the process for $n_0$ updates.  Contrast this with $j \in \mathcal{N}$ whose counter increments if and only if $j$ receives a newer version from one of its neighbors.   Let $X_j(t)$ be the number of versions node $j$ is behind $n_0$ at time $t$, and for any subset $S$ of vertices let $X_S(t) = \min_{i \in S} X_i(t)$.  Then the limiting average version age of $S$ is $v_G(S) = \lim_{t \rightarrow \infty} X_S(t)$.

The information flow from $i$ into set $S$ is denoted by $\lambda_i(S) = \sum_{j \in S} \lambda_i(j) = \frac{\lambda |N(i) \cap S|}{n}$. Similarly $\lambda_0(S) = \frac{\lambda|S|}{n}$.  Then the main result of Yates~\cite{yates_gossip} is that the stochastic quantity $v_G(S)$ can be computed combinatorially by unrolling the following recursion:
\begin{align}\label{eqn:vage}
    v_G(S) = \frac{\lambda_e + \sum_{i \not\in S}\lambda_i(S) v_G(S\cup \{i\})}{\lambda_0(S) + \sum_{i \not\in S}\lambda_i(S)}
\end{align}

In any network where $\lambda$ and $ \lambda_e$ are constant, $\lambda_0(S) = \frac{\lambda |S|}{n}$, and $\lambda_i(j) = \frac{\lambda}{\operatorname{deg}(i)}$, the version age is monotonic and bounded as $\Omega(\log n) = v_G(S) = O(n)$ as shown in Appendix A.  We use this fact numerous times throughout and will attempt to be explicit when it is applied.

\subsection{Random Graphs}
A graph is $d$-regular if all vertices have degree $d$.  We let $G(n,d)$ denote the uniform probability distribution over all $n$-node $d$ regular graphs.  The Erd\H{o}s-Reyni model $G(n,p)$ denotes the distribution on $n$-node graphs where each edge is chosen i.i.d. with probability $p$.  We say a graph property $Q$ holds asymptotically almost surely (a.a.s.) under $G(n,d)$ if $\mathbb{P}[G(n,d) \in Q] \rightarrow 1$ as $n \rightarrow \infty$, and likewise for $G(n,p)$.

A graph property is called monotone if adding edges to a graph does not change the property and the property is invariant to permuting the labels of vertices (e.g. connectivity is monotone, 3-colorability is not). Recall the definition of a random graph threshold:
\begin{definition}\label{def:threshold}
    A function $p^* = p^*(n)$ is a \textit{threshold} for graph property $R$ in $G(n,p)$ if
    \[\lim_{n \rightarrow \infty} \mathbb{P}[G(n,p) \in R] = \begin{cases}
        0 &\text{ if } p/p^* \rightarrow 0\\
        1 &\text{ if } p/p^* \rightarrow \infty
    \end{cases}\]
\end{definition}
A \textit{sharp threshold} is defined similarly, where the limits to 0 and $\infty$ can be bounded by $1- \epsilon$ and $1+ \epsilon$ respectively, for every $\epsilon > 0$. (See eg.~\cite{friedgut_monotone_threshold} for a more formal definition of sharp thresholds.)

\section{Notation and Summary of Results}\label{sec:summary}
Throughout we assume $\lambda_0$ and $\lambda$ are constants and let $n \rightarrow \infty$.  All inequalities are meant to hold for $n$ sufficiently large.  We say a graph property $Q$ holds asymptotically almost surely (a.a.s.) under $G(n,d)$ if $\mathbb{P}[G(n,d) \in Q] \rightarrow 1$ as $n \rightarrow \infty$, and likewise for $G(n,p)$. We use $v_G(S)$ to denote the version age of $S \subseteq \mathcal{N}$, and when there is no ambiguity about the graph the subscript $G$ will be dropped.  The graph $K_{L,R}$ is the complete bipartite graph with vertex bipartition $\mathcal{N} = L \cup R$ and $E(G) = \{uv: u \in L, v \in R\}$. For any $S \subset \mathcal{N}$, the neighborhood of $S$ is $N(S) = \{v \in \mathcal{N}: \operatorname{dist}_G(v, S) = 1\}$ and its edge boundary is $\partial S = \{ij \in E : i \in S, j \not\in S \}$. Finally we use the notation $\Tilde{\Theta}(\cdot)$ to suppress logarithmic factors.   Our main results are presented below.

\begin{theorem}\label{thm:bipartite}Let $L := L(n)$ and $R := R(n)$ be non-decreasing functions such that $|L|+|R| = n$ and for all $n$, $|L| < |R|$. For $K_{L,R}$, if $j \in R$ then,
\begin{enumerate}
    \item $L = \Theta(1) \hspace{2.56cm} \implies \ver{j} = \Theta(n)$
    \item $L = f(n)$ and $f(n) = o(n) \implies \ver{j} = \Omega(n/f(n))$
    \item $L = \nalph$ for $\alpha \in (0,1)  \hspace{.8cm} \implies \ver{j} = \Tilde{\Theta}(\omnalph)$
    \item $L = \Theta(n) \hspace{2.56cm} \implies \ver{j} = \Theta(\log n)$
    \end{enumerate}
    \end{theorem}

\begin{theorem}\label{thm:dreg}
For any fixed $d \geq 3$, a.a.s the worst-case version age of any vertex in $G(n,d)$ is $\Theta(\log n)$.\end{theorem}
\begin{theorem}\label{thm:gnpthresh}
    
Let $\epsilon > 0$. If $p = \frac{(1-\epsilon)\log n }{n}$ then a.a.s. the average version age of a vertex in $G(n,p)$ is $\Omega(n^\epsilon)$. 
 If $p = \frac{(100+\epsilon)\log n}{n}$ then the average version age of a vertex is $\Theta(\log n)$. Furthermore there are constants $\alpha > 0$ and $1 \leq c^* \leq 100$ such that $p = \frac{c^* \log n}{n}$ is a threshold for the graph property ``$G$ has average version age less than $\alpha \log n$''\end{theorem}
\hfill

To summarize in words, the worst-case version age in complete bipartite graphs is inversely related to the size of the smaller component.  This intuitively makes sense, as $K_{0,n} = \overline{K_n}$ is the empty graph and has linear version age, and the balanced bipartite graph $K_{n/2, n/2}$ is very dense so information should flow readily (and we will show has the same version age scaling as $K_n$ up to constant factors).  Theorem~\ref{thm:bipartite} makes this intuition precise.  Our concern with rational partition sizes in $K_{L,R}$ is to compare these graphs to generalized rings with rational degree as studied in Srivastava and Ulukus~\cite{gossip_rings}: in both cases rational degree distribution leads to rational version age scaling.

Perhaps surprising is the contrast between Theorem~\ref{thm:dreg} and Theorem~\ref{thm:gnpthresh}, where $d$-regular graphs have logarithmic version age and $G(n,p)$ only achieves logarithmic version age for expected degree $\sim \log n$.  An explanation for this is that for any fixed $d$, $G(n,d)$ is a very good expander, but $G(n,p)$ only achieves connectivity once $p = \frac{\log n}{n}$ and so below this point is a worst-case expander.  Expansion is a good measure for long-range connectivity in graphs, and connectivity should be a condition for the quick spread of gossip. In fact the proof for Theorem~\ref{thm:dreg} relies heavily on the expansion properties of $G(n,d)$, and likewise the lower bound of Theorem~\ref{thm:gnpthresh} relies on counting isolated vertices in $G(n,p)$.

\section{Bipartite graphs}\label{sec:bipartite}
We prove theorem~\ref{thm:bipartite}.  By symmetry of the complete bipartite graph $K_{L,R}$, the version age of a subset only depends on the number of nodes in the left and right parts.  Therefore for any subset $S \subset V$ with $|S\cap L| = i, |S\cap R| = j$ define $v(i, j)~:=~v_{K_{L,R}}(S)$ and likewise for $\lambda_w(i,j)$.  Finally define $u_{K_{L,R}}(i,j) = \frac{\lambda}{\lambda_e} v_{K_{L,R}}(i,j)$.  
\begin{lemma}\label{lemma:ubipartite}
Let $K_{L,R}$ be a complete bipartite graph on $n$ vertices.  Then for any $S \subset V$ with $S\cap L = i$, $S\cap R = j$, 
\[u_{K_{L,R}(i,j)} = \frac{1 + \frac{(|L|-i)j}{|R|}u(i+1,j) + \frac{(|R|-j)i}{|L|}u(i,j+1)}{\frac{i+j}{n} + \frac{(n-i)j}{|R|} + \frac{(n-j)i}{|L|}}.\]
\end{lemma}
\begin{proof}
    Rearranging equation~\ref{eqn:vage} and using $v(\mathcal{N}) = \frac{\lambda_e}{\lambda}$,
    \[
        v(\mathcal{N}) \lambda =  \lambda_0(S)v(S) + \sum_{i \not\in S}\lambda_i(S)(v(S \cup \{i\}))
    \]
    Note that $u(S) = \frac{v(S)}{v(\mathcal{N})}$, and $\lambda_0(S) = \frac{\lambda |S|}{n}$. By symmetry of the network we can split the sum and simplify,
    \begin{align*}
        1 &= \frac{i+j}{n}u(i,j) + \sum_{(S \cap L)^c}\frac{j}{|R|}(u(i,j) - u(i+1, j))\\
        &\phantom{{}={}}+ \sum_{(S \cap R)^c} \frac{i}{|L|}(u(i,j) - u(i, j+1))\\
        &= u(i,j) \left(\frac{i+j}{n} + \frac{(|L|-i)j}{|R|} + \frac{(|R|-j)i}{|L|} \right)\\
        &\phantom{{}={}}- \frac{(|L|-i)j}{|R|}u(i+i,j) - \frac{(|R|-j)}{|L|}u(i,j+1)
    \end{align*}
    Solving for $u(i,j)$ completes the proof
\end{proof}
\begin{proof}{(Theorem~\ref{thm:bipartite})}
    We split the proof into four parts corresponding to the four regimes in the Theorem.  Since $u(S)$ is a constant multiple of $v(S)$, we are content to bound the former.  If $L = \Theta(1)$, then there are constants $C_1 < \liminf L(n)$ and $C_2 > \limsup L(n)$.  Then for large enough $n$,
    \[\frac{\lambda}{\lambda_e}v(0,1) = u(0,1) \leq \frac{1 + \frac{C_2}{n - C_2} u(1,1)}{\frac{1}{n} + \frac{C_2}{n - C_2}} \leq \frac{n - C_2}{C_2} = \Theta(n)\]
    Where the second inequality uses the fact that for all graphs $u_G(S) > 0$. The lower bound is almost identical, save for switching the roles of $C_1$ and $C_2$.

    For the second case that $L = \frac{f(n)}{n}$ and $f(n) = o(n)$,
    \[u(0,1) \geq \frac{1 + \frac{f(n)}{n - f(n)}u(1,1)}{\frac{1}{n} + \frac{f(n)}{n - f(n)}} = \Omega\left(\frac{n-f(n)}{f(n)}\right) = \Omega\left(\frac{n}{f(n)}\right)\]
    where the last equality follows from the upper bound on  $f$.

    For the third case when $L = \nalph$, we have the lower bound $\Omega(\omnalph)$ from the above observation.  We also have \[u(0,1) = \frac{1 + \frac{\nalph u(1,1)}{n(1-\nalph)}}{\frac{1}{n} + \frac{\nalph}{n(1-\nalph)}} = O\left(\omnalph\right) + O\left(u(1,1)\right)\]
    so it is enough to show $u(1,1) = \Tilde{O}\left(\omnalph\right)$.  To that end we state and prove the following lemma.
    \begin{lemma}\label{lemma:recursion}
        For any complete bipartite graph $K_{L,R}$,
        \[u_{K_{L,R}}(1,1) \leq \min \{|R| \log(|L|), |L| \log(|R|)\}.\]
    \end{lemma}
    \begin{proof}
    For clarity we write $L$ and $R$ instead of $|L|$ and $|R|$ and drop the subscript on $u(\cdot,\cdot)$. By Lemma~\ref{lemma:ubipartite}
        \begin{align*}
            u(k,1) &\leq \frac{1 + \frac{L-k}{R}u(k+1, 1) + \frac{k(R-1)}{L}u(k,2)}{\frac{L-k}{R} + \frac{k(R-1)}{L}}\\
            &\leq \frac{1 + \frac{L-k}{R}u(k+1, 1) + \frac{k(R-1)}{L}u(k,1)}{\frac{L-k}{R} + \frac{k(R-1)}{L}}\\
            &=\frac{RL + L(L-k)u(k+1,1) + kR(R-1)u(k,1)}{L(L-k) + kR(R-1)}
        \end{align*}
    Where the second inequality follows since increasing the size of a set can only decrease its version age.  Letting $D = L(L-k) + kR(R-1)$, we have
    \begin{align*}
    &u(k,1)\left(\frac{L(L-k)}{D}\right) \leq \frac{RL}{D} + \frac{L(L-k)}{D}u(k+1,1)\\
    \implies &u(k,1) \leq \frac{R}{L-k} + u(k+1,1)
    \end{align*}
    Starting at $u(1,1)$ and applying this inequality recursively $L$ times yields $u(1,1) \leq \sum_{i=1}^L \frac{R}{L-i} \leq R \log(L)$.  An analogous argument holds for the quantity $L \log(R)$ by instead expanding $u(1,k)$.
    \end{proof}
    Applying Lemma~\ref{lemma:recursion} completes the proof for the third case.
    
    For the fourth case when $L = \Theta(n)$,  similar to case one we could find constants $C_1, C_2$ so that the tail of $L(n)$ is bounded as $nC_1 \leq L(n) \leq nC_2$.  For clarity we are content to let $L = cn$, and $K(n) := K_{cn, (1-c)n}$ for some absolute constant $c$. We also drop floors and ceilings on $cn$ when this isn't an integer; this obviously doesn't change the asymptotics. We briefly prove the following fact.
    \begin{fact}\label{fact:cleq}
        $u_{K(n)}(1,2) \leq u_{K(2n)}(1,2)$ and $u_{K(n)}(2,1) \leq u_{K(2n)}(2,1)$
    \end{fact}
    \begin{proof}
        When we move from $K_{cn, (1-c)n}$ to $K_{2cn, 2(1-c)n}$, every arc in the larger network has half the capacity of a corresponding arc in the smaller network. Therefore by symmetry of the network any subset $(A,B) \subset K_{2cn, 2(1-c)n}$ will have the same version age scaling as $(A/2, B/2) \subset K_{cn, (1-c)n}$.
    \end{proof}
    Now there is a $\delta = \delta(c) > 0$, and for any $\epsilon > 0$ and $n$ large enough,
    \begin{align*} 
        u_{K(n)}(1,1) &= \frac{1 + \frac{nc - 1}{n(1-c)} u_{K(n)}(2,1) + \frac{n(1-c)-1}{nc}u_{K(n)}(1,2)}{\frac{2}{n} + \frac{nc - 1}{n(1-c)} + \frac{n(1-c)-1}{nc}}\\
         &\leq \frac{1}{\frac{1}{n} + \frac{nc - 1/2}{(1-c)n} + \frac{n(1-c) - 1/2}{nc}} \times\\
         &\phantom{{}={}}\bigg(1 + \frac{nc - 1}{n(1-c)}u_{K(2n)}(2,1)\\
         &\phantom{{}={}}+ \frac{n(1-c) - 1}{nc}u_{K(2n)}(1,2) \bigg) + \epsilon\\
         &\leq \frac{1}{\frac{1}{n} + \frac{nc - 1/2}{(1-c)n} + \frac{n(1-c) - 1/2}{nc}} \times\\
         &\phantom{{}={}}\bigg(1 + \frac{nc - 1/2}{n(1-c)}u_{K(2n)}(2,1)\\
         &\phantom{{}={}}+ \frac{n(1-c) - 1/2}{nc}u_{K(2n)}(1,2) \bigg) + \epsilon - \delta\\
         &= u_{K(2n)}(1,1) + \epsilon - \delta
    \end{align*}
    The equality is by Lemma~\ref{lemma:ubipartite}.  In the first inequality we applied Fact~\ref{fact:cleq} and noted the limits are the same. In the second inequality we used the fact that the version age of a subset in any network is at most linear.  Therefore $u_{K(n)}(1,1) \leq u_{K(2n)}(1,1) + \epsilon - \delta$, setting $\epsilon < \delta$ implies $u_{K(n)}(1,1) = O(\log n)$.  The lower bound follows since the version age of the clique $K_n$ is $\Theta(\log n)$.
\end{proof}

\section{Random Regular Graphs}\label{sec:regular}

The proof for Theorem~\ref{thm:dreg} uses similar techniques to \cite{gossip_grid} for bounding the age of gossip on a grid.  Therein a key part of the argument is understanding expansion properties of the network; how many edges exist in any cut of $G$.  We recall the edge expansion number for a graph (also known as the Cheeger constant, or isoperimetric number).
\begin{definition}
    For any graph $G$, the edge expansion number $h(G)$ is given by
    \[h(G) := \min_{|S| \leq n/2} \frac{|\partial S|}{|S|}\]
    where $\partial S$ is the set of edges in the cut spanning $S$ and $S^c$.
\end{definition}
Recall that $G(n,d)$ is the uniform probability distribution over all $d$-regular graphs.  A result by Bollob\'as~\cite{bollobas_cheeger_d-reg} shows for constant $d$, $G(n,d)$ generates good edge expanders.  We need a weaker version of their result,
\begin{theorem}{Bollob\'as\cite{bollobas_cheeger_d-reg}}\label{thm:bollobas}
    For every fixed $d \geq 3$.  Then there is a constant $c_d < \frac{1}{2}$ such that
    \[\mathbb{P}[h(G(n,d)) \geq d c_d] \rightarrow 1 \text{ as } n \rightarrow \infty \]
\end{theorem}

\begin{proof}{(Theorem~\ref{thm:dreg})}
Recalling identity (4) from \cite{gossip_grid} which is just a rearrangement of equation~\eqref{eqn:vage},
\begin{align}\label{eq:recursion}
\lambda_e = \lambda_0(S)v(S) + \sum_{i \not \in S} \lambda_i(S)(v(S) - V(S \cup \{i\}))
\end{align}
Since $G(n,d)$ is regular, we can partition $\partial S$ into sets $A_1,...,A_d$ where $A_j = \{ v \not\in S : |N(v) \cap S| = j\}$.  Then,
\begin{align*}
    \lambda_e &= \lambda_0(S) v(S) + \sum_{i=1}^d \sum_{j \in A_i}\lambda_j(S)(v(S) - v({S \cup j}))\\
    &\geq \lambda_0(S)v(S) + \sum_{i=1}^{d}|A_i|\frac{\lambda_i}{d} \min _{j \in A_i} (v(S) - v({S \cup j}))\\
    &\geq \lambda_0(S)v(S) + \left(\frac{\lambda}{d}\sum_{i=1}^d i |A_i|\right)(v(S) - \max_{i \in N(S)}(v({S \cup i})))
\end{align*}
Since $\sum_{i=1}^d i |A_i| = \partial S$, by Theorem~\ref{thm:bollobas}, when $|S| < n/2$ a.a.s. $G(n,d)$ satisfies:
\begin{align}\label{eq:subset_small}
    \lambda_e &\geq \frac{\lambda |S|}{n}v(S) + \frac{\lambda}{d}c_d d|S| (v(S) - \max_{i \in N(S)}v({S \cup i}))\nonumber\\
    \implies v(S) &\leq \frac{\frac{\lambda_e}{\lambda} + c_d |S| \max_{i \in N(S)}v_{S \cup i}}{{\frac{|S|}{n} + c_d|S|}}
\end{align}
By an analogous argument for when $|S| > n/2$:
\begin{align}\label{eq:subset_big}
    v(S) &\leq \frac{\frac{\lambda_e}{\lambda} +  c_d (n-|S|) \max_{i \in N(S)}v({S \cup i})}{\frac{|S|}{n} + c_d(n-|S|)}
\end{align}
    Therefore when unrolling recursion~\eqref{eq:recursion} for $v(\{i\})$, if $S < n/2$ we use inequality~\eqref{eq:subset_small}, otherwise we use~\eqref{eq:subset_big}.  To that end let $X$ be the sum corresponding to small subset size and letting $j := |S|$,
\begin{align*}
    X &\leq \frac{\lambda_e}{\lambda}\left(\frac{1}{c_d + \frac{1}{n}}\right) \left(1 + \sum_{i=1}^{n/2} \prod_{j=1}^{i} \frac{c_d j}{\frac{j+1}{n} + c_d(j+1)} \right)\\
    &\leq \frac{\lambda_e}{c_d\lambda}\left(1 + \sum_{i=1}^{n/2} \prod_{j=1}^{i} \frac{c_d j}{\frac{j+1}{n} + c_d(j+1)} \right)\\
    &\leq \frac{\lambda_e}{c_d\lambda}\left(1 + \sum_{i=1}^{n/2} \prod_{j=1}^{i} \frac{j}{j+1} \right)\\
    & = \frac{\lambda_e}{c_d\lambda}\left(1 + \sum_{i=1}^{n/2} \frac{1}{i+1} \right) = O(\log n)
\end{align*}
    Letting $Y$ be the terms corresponding to $|S| > n/2$, it can be shown that
\begin{align*}
    Y &\leq \frac{\lambda_e}{\lambda} + \frac{\lambda_e}{\lambda}\left(1 + \sum_{i = n/2}^{n-2} \prod_{j=n/2}^{i} \frac{c_d (n-j)}{\frac{j}{n} + c_d (n-j-1)} \right)\\
    &\phantom{{}={}}\times \prod_{j=1}^{n/2-1} \frac{c_d j}{\frac{j}{n} + c_d(j+1)} \times \frac{1}{1/2 + c_dn/2}\\
    &\leq C' \frac{\lambda_e}{\lambda} \log n = O(\log n)
\end{align*}
Where we omitted computations that are analogous to those found in the previous step and in~\cite{gossip_grid}.  Finally noting that for any network $v(\{i\}) = \Omega( \log n)$ completes the proof.
\end{proof}

\section{Erd\H{o}s-Reyni Random Graphs}\label{sec:gnp}
We need the following lemmas on properties of $G(n,p)$,
\begin{lemma}\label{lem:thresholds}
    The following holds a.a.s.
    \begin{enumerate}
        \item If $d < 1$ and $p = d/n$, then $G(n,p)$ has no component of size larger than $O(\log n)$
        \item If $\frac{1}{2} < d < 1$ and $p = \frac{d \log n}{n}$ then $G(n,p)$ has $\Omega(n^{1-d})$ isolated vertices
        \item If $d > 100$ and $p = \frac{d \log n}{n}$ then $G(n,p)$ is connected and for any $\epsilon > 0$, every $v \in G(n,p)$ satisfies $\operatorname{deg}(v) \in \left((1-\epsilon)np, (1+\epsilon)np \right)$
    \end{enumerate}
\end{lemma}
\begin{proof}
    Items one and three are easily verified (\cite{blum_fds}, \cite{gnp_phase_easy}). We leave Item two for Appendix B.
\end{proof}
\begin{lemma}\label{lem:subsetregularity}
Let $0 < \delta < 1$ and $p \geq \frac{30 \log n}{\delta^2 n}$.  Then a.a.s. any subset $S$ of $G(n,p)$ with $|S| < \frac{n}{2}$ satisfies 
\[\mathbb{P}\left( \left| \partial S \right| \not\in (1 \pm \delta)\mathbb{E}[\left|\partial S \right|]\right) = o(1)\]
\end{lemma}
\begin{proof}
    For any $k = |S| \subset \mathcal{N}$ and $\alpha = \delta\sqrt{k(n-k)p}$, define $A_S := \mathbb{P}\big[|\partial S| \not\in (1 \pm \delta)\mathbb{E}[|\partial S|]\big]$. Since in $G(n,p)$ edges are included $i.i.d$, the Chernoff bound states
    \[\mathbb{P}[A_S] \leq 3 \exp\left(\frac{-\alpha^2}{8}\right).\]
    then by the choice of $p$, for $n$ large,
    \begin{align*}
    \frac{9 \log \left(n{n \choose k}\right)}{\delta^2 k(n-k)}\leq \frac{9 \log(\frac{ne}{k})}{\delta^2(n-k)} + \frac{9\log (n)}{\delta^2 (n-1)}\leq \frac{9(1 + \log n)}{\delta^2 \frac{n}{2}} + \frac{10 \log n}{\delta^2 n} \leq p
    \end{align*}
    So that by the definition of $\alpha$,
    \begin{align*}
        \mathbb{P}\left(\bigcup_{\substack{S \subset \mathcal{N}\\ |S| \leq n/2}}A_S\right) &\leq 3\sum_{k=1}^{n/2} {n \choose k}\exp\left(\frac{-\delta^2 k(n-k)p}{8}\right) \leq 3\sum_{k=1}^{n/2} o\left(\frac{1}{n} \right) = o(1)
    \end{align*}
\end{proof}

\begin{proof}{(Theorem~\ref{thm:gnpthresh})}
    Let $\epsilon  > 0$ and suppose $p \leq \frac{(1-\epsilon)\log n}{n}$.  Then by Lemma~\ref{lem:thresholds}.2, there are $O(n^\epsilon)$ isolated vertices.  By equation~\eqref{eqn:vage} any isolated vertex $i$ has version age $v(\{i\}) = \frac{\lambda_e n}{\lambda} = \Theta(n)$.  For any $i$ that is not isolated,
    \[v_{G(n,p)}(\{i\}) = \Omega(\log n).\]  Then the average version age of a vertex is greater than \[\frac{1}{n}(\Theta\left(n^{1 + \epsilon}\right) + \Omega\left((n - n^\epsilon)\right) \log n) = \Omega(n^\epsilon).\]

    Now let $p \geq \frac{100 \log n}{n}$. Define $A_j = \{ v \in V(G) : |N(v) \cap S| = j\}$ and $u(S) = \frac{v(S)}{v(\mathcal{N})}$.  Applying Lemma~\ref{lem:thresholds}.3 with $\epsilon = 1/2$ and Lemma~\ref{lem:subsetregularity} with $\delta = 1/3$, a.a.s. for every $k = |S| \subset \mathcal{N}$ we have 
    \begin{align*}
        v(\mathcal{N}) \lambda &=  \lambda_0(S)v(S) + \sum_{i \not\in S}\lambda_i(S)(v(S \cup \{i\}))\nonumber\\
        &= \frac{\lambda k}{n} v(S) + \sum_{i \not\in S} \frac{|N(i) \cap S|}{\operatorname{deg}(i)}(v(S) - V(S \cup \{i\}))\nonumber
    \end{align*}
    Where the first line is identity~\eqref{eq:recursion}. Then a.a.s.
    \begin{align}
        n &\geq k u(S) + \frac{2n}{3 \log n} \sum_{\ell <  3\log n}A_\ell (u(S) - u(S \cup \{i\}))\label{eq:gnp1}\\
        &\geq k u(S) + \frac{n}{2 \log n}\times  \frac{k (n-k) \frac{2}{3} \log n}{n}\times (u(S) - \max_{i \in N(S)} u(s \cup{i}))\label{eq:gnp2}\\
        &\geq  k u(S) + \frac{1}{3}k(n-k)(u(S) - \max_{i \in N(S)} u(s \cup{i}))\nonumber
    \end{align}
    Where~\eqref{eq:gnp1} uses Lemma~\ref{lem:thresholds} and connectivity, and~\eqref{eq:gnp2} uses Lemma~\ref{lem:subsetregularity} and symmetry. Rearranging, 
    \begin{align}
        u(S) \leq \frac{n + \frac{1}{3}k(n-k) \max u(S \cup \{i\})}{k + \frac{1}{3}k(n-k)}\label{eq:gnp4}
    \end{align}
    Let $P_i = \prod_{j=i}^{n-1}\frac{\frac{1}{3}j(n-j)}{j + \frac{1}{3}j(n-j)}$. Expanding a single vertex $\ell \in \mathcal{N}$,
    \begin{align*}
        u(\{\ell\}) &\leq \frac{n}{1 + \frac{n-1}{3}} + \frac{\frac{1}{3}(n-1)}{1 + \frac{1}{3}(n-1)}+ \sum_{i=1}^n \frac{n}{(i+1) + \frac{1}{3}(i+1)(n-i-1)}P_i
    \end{align*}
    Observing that $P_i = \prod_{j=i}^{n-1} \frac{n-j}{n-j+3} \leq \prod_{j=i}^{n-1} \frac{n-j}{n-j+1} = \frac{1}{n-i+1}$,
    \begin{align}\label{eq:gnp_final}
        u(\{\ell\}) &\leq C_1 + 3n\sum_{i=1}^{n} \frac{1}{i(n-i + \frac{1}{3})}
    \end{align}
    Where the $\frac{1}{3}$ is present to avoid diving by zero and facilitate the next step.  Briefly recall some properties of the Digamma function~\cite{digamma}: for every complex $z$ that is not a negative integer, the Digamma function has series representation $\psi(1 + z) = - \gamma + 
    \sum_{k=1}^{\infty} \frac{z}{k(k+z)}$, where $\gamma$ is the Euler–Mascheroni constant.  Also $\psi(1-x) = \psi(x) + \pi \cot (\pi x)$ and $\psi(1+x) \sim \log x$ if $x \in (0,\infty)$.  Therefore ~\eqref{eq:gnp_final} is bounded as,
    Therefore ~\eqref{eq:gnp_final} is bounded as,
    \begin{align*}
        u(\{\ell\}) &\leq C_1 +  \sum_{i=1}^\infty \frac{3n}{i(n-i + \frac{1}{3})} + \sum_{i=n+1}^\infty\frac{3n}{i(i-n - \frac{1}{3})}\\
        &\leq C_2 + 3\sum_{i=1}^\infty \frac{-(n + \frac{1}{3})}{i(i-(n+\frac{1}{3}))} + 3\sum_{i=1}^\infty \frac{n}{i(n+i)}\\
        &\leq C_3 + \psi(1-(n+\frac{1}{3})) + \psi(n+1)\\
        &= O(\log n)
    \end{align*}

    We conclude by remarking on the existence of a threshold for super-logarithmic to logarithmic version age in $G(n,p)$. Fix $\epsilon > 0$, $\delta > 0$ and the sequences $p_0(n) := \frac{(1-\epsilon) \log n}{n}$ and $p_1(n) := \frac{(100 + \epsilon)\log n}{n}$. Consider the graph property $P= $`$G$ has average version age less than $\alpha \log n$'.  We just proved that $\mathbb{P}[G(n,p_0) \in P] \leq \delta$, and $\mathbb{P}[G(n,p_1) \in P] \geq 1 -\delta$, when $\alpha$ is suitably large.  Since adding edges to any $G$ can only decrease the average version age (Appendix A), $P$ is a monotone graph property invariant under vertex permutations.  Therefore for all $n$, there exists a $p^*(n) \in (p_0(n), p_1(n))$ where $\mathbb{P}[G(n,p^*) \in P] = \frac{1}{2}$.  By the shrinking critical window, $\lim_{n \rightarrow \infty} \frac{p_1(n) - p_0(n)}{p_c(n)}$ exists, so definition~\ref{def:threshold} is satisfied and a threshold exists.
\end{proof}
We briefly note that part 1 of Lemma~\ref{lem:thresholds} implies for $d < 1$ and $p < d/n$, the expected version age of a vertex in $G(n,p)$ is $\Omega\left(\frac{n}{\log n} \right)$.  For $p$ much greater than in part 3, the version age still scales logarithmically, but the constant factors decrease.

\section{Remarks}
While we have roughly characterized the behavior of vAoI in $G(n,p)$, the question of identifying the precise scaling of the threshold is still open. We haven't made much effort to optimize the constants in Theorem~\ref{thm:gnpthresh}.  More work is needed to show the existence of a sharp threshold, and the author believes this will require novel structural results for $G(n,p)$ at or slightly above the sharp connectivity threshold.

There are many open problems left for the study of randomness in gossip networks.  For instance, the 2-d lattice has version age scaling as $n^{1/3}$~\cite{gossip_grid}.  What would happen if we considered bond percolation on the lattice instead?  How would the version age scale with $p$?  There also exist random spatial graph models, such as preferential attachment and nearest neighbor attachment which could model real-world networks in a more realistic way. The combinatorial nature of the vAoI may make studying these models tractable.

\section*{Acknowledgements}
    This research was funding in part by NSF grant  ECCS-2217023.  I would like to thank Sennur Ulukus and Arunabh Srivastava for helpful discussions on section~\ref{sec:regular}, and Xing Gao for help with details in Theorem~\ref{thm:dreg}. Thanks to Lev Reyzin for pointing out a need for Theorem 5.

\newpage
\bibliographystyle{IEEEtran}
\bibliography{bibliography}

\newpage
\onecolumn
\section*{Appendix A}\label{app:monotone}
We show the vAoI function is monotone.  We assume $\lambda, \lambda_e$ are constants, and $\lambda_{i}(j) = \frac{\lambda}{\operatorname{deg}(i)}$.
\begin{lemma}
    For any graph $G$, let $G'$ be obtained from $G$ by adding an edge $(uv)$.  Then for any $S \subset G$, \[v_{G'}(S) \leq v_G(S)\]
\end{lemma}
\begin{proof}
    Recalling equality~\ref{eqn:vage}, note that if $u \in S$ and $v \in S$, then the version age of $S$ in $G$ and $G'$ are identical.  Therefore we only need to consider the cases when one of $u,v$ are in $S$, or neither are in $S$. We prove these cases separately via reverse induction on $S \subset G$. To that end, the base case is $S = \mathcal{N}$ and the statement is true by the observation above.  For the inductive hypothesis, suppose the claim is true for $|S'| = k > 1$.  We now consider sets $|S| = k-1$.

    Case 1 is that $u \in S$ but $v \not\in S$.  Then,

    \begin{align*}
        v_{G'(S)} &= \frac{\frac{\lambda_e}{\lambda} + \sum\limits_{i \not\in S \cup \{v\}} \left(\frac{|S \cap N_G(i)|}{\operatorname{deg}_{G'}(i)}v_{G'}(S \cup \{i\})\right) + \frac{|S \cap N_{G'}(v)|}{\operatorname{deg}_{G'}(v)}v_{G}(S \cup \{v\})}{\frac{|S|}{n} + \sum\limits_{i \not\in S \cup \{v\}}\left( \frac{|S \cap N_G(i)|}{\operatorname{deg}_{G'}(i)}\right) + \frac{|S \cap N_{G'}(v)|}{\operatorname{deg}_{G'}(v)}}\\
        &\leq  \frac{\frac{\lambda_e}{\lambda} + \sum\limits_{i \not\in S \cup \{v\}} \left(\frac{|S \cap N_G(i)|}{\operatorname{deg}_{G'}(i)}v_{G}(S \cup \{i\})\right) + \frac{|S \cap N_{G'}(v)|}{\operatorname{deg}_{G'}(v)}v_{G}(S \cup \{v\})}{\frac{|S|}{n} + \sum\limits_{i \not\in S \cup \{v\}}\left( \frac{|S \cap N_G(i)|}{\operatorname{deg}_{G'}(i)}\right) + \frac{|S \cap N_{G'}(v)|}{\operatorname{deg}_{G'}(v)}}\\
        &\leq \left(\frac{\frac{|S|}{n} + \sum\limits_{i \not\in S \cup \{v\}}\left( \frac{|S \cap N_G(i)|}{\operatorname{deg}_{G'}(i)}\right)}{\frac{|S|}{n} + \sum\limits_{i \not\in S \cup \{v\}}\left( \frac{|S \cap N_G(i)|}{\operatorname{deg}_{G'}(i)}\right) + \frac{|S \cap N_{G'}(v)|}{\operatorname{deg}_{G'}(v)}}\right) v_G(S)\\
        &\phantom{{}={}}+ \left(\frac{\frac{|S \cap N_{G'}(v)|}{\operatorname{deg}_{G'}(v)}}{\frac{|S|}{n} + \sum\limits_{i \not\in S \cup \{v\}}\left( \frac{|S \cap N_G(i)|}{\operatorname{deg}_{G'}(i)}\right) + \frac{|S \cap N_{G'}(v)|}{\operatorname{deg}_{G'}(v)}} \right)v_G(S \cup \{v\})\\
        &\leq \left(\frac{\frac{|S|}{n} + \sum\limits_{i \not\in S \cup \{v\}}\left( \frac{|S \cap N_G(i)|}{\operatorname{deg}_{G'}(i)}\right)}{\frac{|S|}{n} + \sum\limits_{i \not\in S \cup \{v\}}\left( \frac{|S \cap N_G(i)|}{\operatorname{deg}_{G'}(i)}\right) + \frac{|S \cap N_{G'}(v)|}{\operatorname{deg}_{G'}(v)}}\right) v_G(S)\\
        &\phantom{{}={}}+ \left(\frac{\frac{|S \cap N_{G'}(v)|}{\operatorname{deg}_{G'}(v)}}{\frac{|S|}{n} + \sum\limits_{i \not\in S \cup \{v\}}\left( \frac{|S \cap N_G(i)|}{\operatorname{deg}_{G'}(i)}\right) + \frac{|S \cap N_{G'}(v)|}{\operatorname{deg}_{G'}(v)}} \right)v_G(S)\\
        &= v_G(S)
    \end{align*}
    Where the first inequality is from the inductive hypothesis, the second is multiplying by 1 and applying equation~\ref{eqn:vage}, the third is the fact that adding a vertex to a set cannot increase the version age, and the fourth is just rearranging the coefficients in the large brackets.

    Case 2 is when neither $u$ or $v$ are in $S$. Notice that the only terms obtained by unrolling equation~\ref{eqn:vage} that are different in $G$ and $G'$ are those when either $u$ or $v$ are neighbors of the $S$; otherwise the extra edge $uv$ doesn't play a role in the expansion, or the edge is present in $S$ so the terms are identical. Therefore we only need to consider those sets which are distance 1 from $\{u,v\}$.  For completeness we can then perform analogous computations to above.
    \begin{align*}
        &v_{G'}(S) =\\
        &\phantom{{}={}}\frac{\frac{\lambda_e}{\lambda} + \sum\limits_{i\in N_{G'}(S) \setminus \{u\} \setminus \{v\}} \left(\frac{|S \cap N_G(i)|}{\operatorname{deg}_{G'}(i)}v_{G'}(S \cup \{i\})\right) + \frac{|S \cap N_{G'}(v)|}{\operatorname{deg}_{G'}(v)}v_{G'}(S \cup \{v\}) + \frac{|S \cap N_{G'}(u)|}{\operatorname{deg}_{G'}(u)}v_{G'}(S \cup \{u\})}{\frac{|S|}{n} + \sum\limits_{i\in N_{G'}(S) \setminus \{u\} \setminus \{v\}}\left( \frac{|S \cap N_G(i)|}{\operatorname{deg}_{G'}(i)}\right) + \frac{|S \cap N_{G'}(v)|}{\operatorname{deg}_{G'}(v)} + \frac{|S \cap N_{G'}(u)|}{\operatorname{deg}_{G'}(u)}}\\
        &\leq  \frac{\frac{\lambda_e}{\lambda} + \sum\limits_{i\in N_{G'}(S) \setminus \{u\} \setminus \{v\}} \left(\frac{|S \cap N_G(i)|}{\operatorname{deg}_{G'}(i)}v_{G}(S \cup \{i\})\right) + \frac{|S \cap N_{G'}(v)|}{\operatorname{deg}_{G'}(v)}v_{G}(S \cup \{v\}) + \frac{|S \cap N_{G'}(u)|}{\operatorname{deg}_{G'}(u)}v_{G}(S \cup \{u\})}{\frac{|S|}{n} + \sum\limits_{i\in N_{G'}(S) \setminus \{u\} \setminus \{v\}}\left( \frac{|S \cap N_G(i)|}{\operatorname{deg}_{G'}(i)}\right) + \frac{|S \cap N_{G'}(v)|}{\operatorname{deg}_{G'}(v)} + \frac{|S \cap N_{G'}(u)|}{\operatorname{deg}_{G'}(u)}}\\
        &\leq \left(\frac{\frac{|S|}{n} + \sum\limits_{i\in N_{G'}(S) \setminus \{u\} \setminus \{v\}}\left( \frac{|S \cap N_G(i)|}{\operatorname{deg}_{G'}(i)}\right)}{\frac{|S|}{n} + \sum\limits_{i\in N_{G'}(S) \setminus \{u\} \setminus \{v\}}\left( \frac{|S \cap N_G(i)|}{\operatorname{deg}_{G'}(i)}\right) + \frac{|S \cap N_{G'}(v)|}{\operatorname{deg}_{G'}(v)} + \frac{|S \cap N_{G'}(u)|}{\operatorname{deg}_{G'}(u)}}\right) v_G(S)\\
        &\phantom{{}={}}+ \left(\frac{\frac{|S \cap N_{G'}(v)|}{\operatorname{deg}_{G'}(v)}}{\frac{|S|}{n} + \sum\limits_{i\in N_{G'}(S) \setminus \{u\} \setminus \{v\}}\left( \frac{|S \cap N_G(i)|}{\operatorname{deg}_{G'}(i)}\right) + \frac{|S \cap N_{G'}(v)|}{\operatorname{deg}_{G'}(v)} + \frac{|S \cap N_{G'}(u)|}{\operatorname{deg}_{G'}(u)}} \right)v_G(S \cup \{v\})\\
        &\phantom{{}={}}+ \left(\frac{\frac{|S \cap N_{G'}(u)|}{\operatorname{deg}_{G'}(u)}}{\frac{|S|}{n} + \sum\limits_{i\in N_{G'}(S) \setminus \{u\} \setminus \{v\}}\left( \frac{|S \cap N_G(i)|}{\operatorname{deg}_{G'}(i)}\right) + \frac{|S \cap N_{G'}(v)|}{\operatorname{deg}_{G'}(v)} + \frac{|S \cap N_{G'}(u)|}{\operatorname{deg}_{G'}(u)}} \right)v_G(S \cup \{u\})\\
        &\leq v_G(S)
    \end{align*}
\end{proof}
We have the easy corollary, which follows from above and results by~\cite{yates_gossip}.
\begin{corr}
    For any graph $G$ and any $i \in G$,
    \[\Omega(\log n) \leq v_G(\{i\}) \leq O(n)\]
\end{corr}

\section*{Appendix B}
We prove Lemma~\ref{lem:thresholds}.2

\begin{proof}
    By linearity of expectation, the expected number of isolated vertices $\mu$ is
    \[n(1-p)^{n-1} \approx n\left(1- \frac{d\log n}{n}\right)^n \approx n \exp (-d \log n) =\frac{n}{n^d} = n^{1-d}.\]
    Let $I_i$ be the indicator r.v. that vertex $i$ is isolated.  Then,
    \begin{align*}
        \mathbb{E}\left[\sum_{i \in \mathcal{N}}I_i \sum_{j \in \mathcal{N}} I_j \right]\ &= \mathbb{E}\left[\sum_{i,j \in \mathcal{N}}I_iI_j\right]\\
        &= \mathbb{E}\left[\sum_{i \not=j}I_{ij}\right] + \mathbb{E}\left[\sum_{i=j}I_{ij}\right]\\
        &= \mu + 2{n \choose 2} (1-p^{2n-3})\\
        &\approx \mu + 2{n \choose 2} \exp\left(\frac{-d\log n}{n}(2n-3)\right)\\
        &= \mu + n(n-1)n^{-2d + \frac{3d}{n}}\\
        &= \mu + (n-1)n^{1-2d + \frac{3d}{n}}\\
        &\leq \mu + \frac{n}{2}
    \end{align*}
    Where the last inequality is for $n$ large enough, since $d > 1/2$.  Let $X$ be the r.v. counting the number of isolated vertices.  We have asymptotically $\operatorname{Var}(X) \leq \frac{n}{2} + n^{1-d} - n^{2(1-d)} \leq n$.  Then by Chebyshev's inequality,
    \[
        \mathbb{P}\left[\left|n^{1-d} - X\right| \geq n^d\right] \leq \frac{\operatorname{Var}(X)}{n^{2d}} \leq \frac{n}{n^{2d}} = o(1)
    \]
    So,
    \[\mathbb{P}\left[X \leq n^{1-d} - n^d\right] = o(1)\]
    Therefore a.a.s. for our choice of $p$, there are at least $n^{1-d} - n^d = O\left(n^{1-d}\right)$ isolated vertices in $G(n,p)$.
\end{proof}

\end{document}